\documentclass[A4paper ,10pt]{article}

\usepackage{amsmath, amsthm, amsfonts}
\usepackage[a4paper, total={6in, 8in}]{geometry}
\usepackage[hidelinks]{hyperref}
\usepackage{graphicx}

\newtheorem{Theorem}{Theorem}
\newtheorem{Corollary}{Corollary}

\title{Entropic Uncertainty Relations via Direct-Sum Majorization Relation for Generalized Measurements}
\author{Kyunghyun Baek $^{1,2,}$*
	, Hyunchul Nha $^{1,}$* and Wonmin Son $^{3,}$*\\
\small $^{1}$  Department of Physics, Texas A$\&$M University at Qatar, Education City, P.O. Box 23874, {Doha}, Qatar\\
\small $^{2}$   Asia Pacific Center for Theoretical Physics, Pohang 37673, Korea\\
\small $^{3}$   Department of Physics, Sogang University, Mapo-gu, Shinsu-dong, Seoul 121-742, Korea\\
\small * baek1013@gmail.com (K.B.); hyunchul.nha@qatar.tamu.edu (H.N.);
sonwm@physics.org (W.S.)\\
}
\date{\vspace{-5ex}}

\begin{document}
\maketitle

\abstract{We derive an entropic uncertainty relation for generalized positive-operator-valued measure (POVM) measurements via a direct-sum majorization relation using Schur concavity of entropic quantities in a finite-dimensional Hilbert space. Our approach provides a significant improvement of the uncertainty bound compared with previous majorization-based approaches [S. Friendland, V. Gheorghiu and G. Gour, Phys. Rev. Lett. 111, 230401 (2013); A. E. Rastegin and K. \.Zyczkowski,  J. Phys. A, 49, 355301 (2016)], {particularly by extending the direct-sum majorization relation first introduced in [\L. Rudnicki, Z. Pucha\l{}a and K.  \.{Z}yczkowski, Phys. Rev. A 89, 052115 (2014)]}. We illustrate the usefulness of our uncertainty relations by considering a pair of qubit observables in a two-dimensional system and randomly chosen unsharp observables in a three-dimensional system. We also demonstrate that our bound tends to be stronger than the generalized Maassen--Uffink bound with an increase in the unsharpness effect. Furthermore, we extend our approach to the case of multiple POVM measurements, thus making it possible to establish entropic uncertainty relations involving more than two observables.}
\\


\section{Introduction}

Ever since Heisenberg introduced the uncertainty principle \cite{Heisenberg1927}, it has laid at the heart of quantum physics as one of the fundamental principles manifesting a profound distinction between quantum and classical physics. Early formulations of uncertainty relations (URs) were made on the basis of statistical variance by Kennard \cite{Kennard1927}, Weyl \cite{Weyl1928}, and Robertson \cite{Robertson1929}. These variance-based URs clearly indicate an inherent limitation to preparing a quantum state with a narrow distribution in both position and momentum observables simultaneously.  In addition, they provided a useful insight into developing URs in terms of other quantum state statistical characteristics, such as mixedness \cite{Dodonov2002} and non-Gaussianity \cite{Baek2018,Son2015,Mandilara2012}, {and into developing entanglement criteria for general quantum systems {\cite{Nha2006,Nha2008,Sun2009,Nha2012}}}.

In finite-dimensional Hilbert spaces, however, D. Deutsch pointed out a drawback of Robertson's bound due to its state-dependent {nature. Indeed, Robertson's bound} may even vanish for certain quantum states with noncommuting observables, thus yielding no meaningful uncertainty relation~\cite{Deutsch1983}. Alternatively, he derived the entropic uncertainty relation (EUR) by using Shannon entropy as an information-theoretical measure of uncertainty.
His seminal work was further improved with the Maassen--Uffink EUR \cite{Uffink1988} following {Kraus'} conjecture \cite{Kraus1987}. {This EUR was} subsequently extended to the case of generalized measurements \cite{Krishna2002}. {Also, it was generalized to} general entropy functions, such as those of Tsallis \cite{Rastegin2013} and R\'enyi \cite{Bialynicki-Birula2006}. Another important advantage to using the information-theoretic approach is that the entanglement effect can be incorporated into the uncertainty paradigm by introducing the concept of quantum memory \cite{Berta2010,Tomamichel2011,Coles2012}. Those EURs form crucial key elements in {detecting entanglement and} proving the security of quantum cryptography, as extensively reviewed in {\mbox{\cite{Wehner2010,Coles2017,Toscano2018}}}. {More recently, it has been discovered that the EURs with quantum memory allow for trade-offs between the concepts of quantum uncertainty and reality for quantum observables \mbox{\cite{Rudnicki2018}}.}

Despite the successful formulation of the uncertainty principle via entropy functions, one may ask whether those specific functions are the ultimate measure of uncertainty.
Beyond specific functions quantifying the degree of uncertainty, URs that are universally applicable to any appropriate uncertainty functions were introduced by using the concept of majorization in \cite{Partovi2011,Friedland2013,Puchala2013}. This approach can be briefly described as follows.
For a pair of probability vectors ${p}$ and $q$, if one can obtain ${p}$  {by making a doubly stochastic matrix $S$ act on $q$}, i.e., ${p}= S q,$
where $S$ is a square matrix whose elements are positive values satisfying $\sum_{i}S_{ij}=\sum_{j}S_{ij}=1$,  $p$ is said {to be majorized by $q$. This is expressed as {\cite{Marshall2011}}}
\begin{align}\label{Maj}
{p}\prec q.
\end{align}

In this case, one may say that ${p}$ is more uncertain than ${q}$, since {the action of a doubly stochastic matrix always makes} a probability distribution more equally distributed. Thus, if a function $f$ is a legitimate measure of uncertainty, it should preserve the partial order indicated by the majorization relation,  i.e., {$f(p)\geq f(q)$} \cite{Uffink1990}, such as R\'enyi and Tsallis entropies.
This majorization-based UR provides universal applicability to any appropriate uncertainty functions with such an uncertainty-order preserving property. {Besides uncertainty relations}, {the concept of majorization is applied to various topics, such as quantum thermodynamics \mbox{\cite{Horodecki2018}} and coherence \mbox{\cite{Zhu2017}}.}

The majorization-based UR was first derived on the basis of the tensor-product majorization \mbox{relation~\cite{Friedland2013,Puchala2013}}. Subsequently, it was applied to {the} direct-sum majorization relation {for rank-1 projective measurements} in {{\cite{Rudnicki2014}}}, providing {stronger bounds for the sum of two entropies} than the former one{, and extended to projection-valued measures in {\cite{Rudnicki2015}}}. Its extension to generalized measurements was also investigated in the tensor-product majorization relation \cite{Friedland2013} and, more recently, in the direct-sum formulation \cite{Rastegin2016}.
However, unlike the case of projective measurements,  {there has not been an extensive examination of whether the direct-sum majorization still provides stronger bounds than the tensor-product one for unsharp positive-operator-valued measure (POVM) measurements.}
In this paper, we propose a new generalization of the direct-sum majorization relation to general POVM measurements. {As the direct-sum majorization relation provides stronger bounds for the case of projective measurements \mbox{\cite{Rudnicki2014,Rudnicki2015,Rastegin2016}}, we show that for general POVM measurements, our generalization improves upon the previously established bounds found in the literature.} We illustrate it by considering a pair of qubit observables in two-dimensional systems and also randomly chosen observables in three-dimensional systems through extensive numerical calculations.

This paper is organized as follows. In Section \ref{sec2},  we briefly introduce the basic concepts and terminologies necessary for our work. {We further} review recent results on majorization-based URs, with a particular focus on the case of generalized measurements. In Section \ref{sec3}, we obtain a direct-sum majorization relation for general POVM measurements and subsequently establish EURs in terms of R\'enyi and Tsallis entropies, including the Shannon entropy. In Section \ref{sec4}, we illustrate the power of our approach by comparing our bound with other known bounds using observables in two-dimensional and three-dimensional systems.  In Section \ref{sec5}, we further extend our approach to obtain a direct-sum majorization relation involving multiple POVM measurements, and we establish the corresponding EURs.

\section{Preliminaries}\label{sec2}



A generalized measurement $A$ can be described by a positive-operator-valued measure (POVM), which is a set of positive operators $\{\hat A_i\}_{i=1}^{n_A}$ satisfying the completeness relation, $\sum_{i=1}^{n_A} \hat A_i=\hat I$, where $n_A$ is the number of different outcomes. In a general scenario in which a quantum state described by a density operator $\hat \rho$ is measured by $A$, {the} probability to obtain the $i$th outcome is given by $$p_i^A=\text{Tr}[\hat\rho \hat A_i].$$
If all elements of a POVM $A$ are orthogonal {to} each other, i.e., $\hat A_i \hat A_j=\delta_{ij}\hat A_i $, or, equivalently, each element is given by a projection, then it is called a projection-valued measure (PVM). Furthermore, in {the} most ideal case, a set of projections provides orthogonal bases, and it is referred to {as} rank-1 PVM.


In an information-theoretic approach, the amount of uncertainty induced by a generalized measurement can be quantified using entropic quantities, such as R\'enyi and Tsallis entropies.
The R\'enyi entropy is defined as
\begin{align}\label{Renyi}
H_{\alpha}(p)=\frac{1} {1-\alpha} \ln\sum_i p_i^\alpha
\end{align}
for $\alpha>0$ with $\alpha\neq1$. In the limit $\alpha\rightarrow 1$, it reduces to the Shannon entropy $H(p)=-\sum_i p_i\ln p_i$. We note that the R\'enyi entropy monotonically decreases with respect to the order $\alpha$.
The Tsallis entropy is also defined for $\alpha>0$, $\alpha\neq1$, as
\begin{align}\label{Tsallis}
T_{\alpha}(p)=\frac{1} {1-\alpha} \left(\sum_i p_i^\alpha-1\right).
\end{align}
Similar to R\'enyi entropies, the Tsallis entropy corresponds to the Shannon entropy  at $\alpha=1$.

{Now, let us introduce} {an equivalent	way to define the majorization relation in Equation {\eqref{Maj}} by means of a set} {of inequalities, which is more useful in the derivation of our results.}
{Suppose that the probability vector ${p}^A_\downarrow=(p_{[1]}^A,p_{[2]}^A,...,p_{[n]}^A)^T$ denotes the rearrangement of $p^A=(p_1^A,p_2^A,...,p_n^A)$ in decreasing order, i.e., $p_{[1]}^A\geq p_{[2]}^A\geq...\geq p_{[n]}^A$, and likewise for $p^B_\downarrow$.}
If they satisfy {{\cite{Marshall2011}}}
\begin{align}\label{MajIneq}
\sum_{i=1}^{k} p_{[i]}^A \leq \sum_{i=1}^{k} p_{[i]}^B
\end{align}
for all $1\leq k\leq n$, {along} with the normalization condition, ${p}^A$ is said to be majorized by ${p}^B$, expressed as
$
{p}^A\prec {p}^B.
$
 {Observe that in order to have the majorization $p^A\prec p^B$, it is enough that $\sum_{i=1}^k p_{[i]}^A \leq \sum_{i=1}^k p_{i}^B$ for any $k$; i.e., in Equation {\eqref{MajIneq}}, the ordered components $p_{[i]}^B$ can be replaced by the unordered ones $p_i^B$, since $\sum_{i=1}^k p_{i}^B \leq \sum_{i=1}^k p_{[i]}^B$.}
As noted earlier, an {appropriate} uncertainty function should give {a smaller value} for ${p}^B$.
 {Schur concave functions are the class of functions preserving this order.}
We note that both R\'enyi and Tsallis entropies are Schur concave, thus preserving the partial order induced by majorization. By utilizing Schur concavity of entropic quantities, one can derive {EURs} from the majorization relation---the so-called majorization EURs.

Majorization EURs for generalized measurements were established first on the basis of the tensor-product majorization relation. For probability vectors ${p}^A$ and ${p}^B$ associated with POVMs $A=\{\hat A_i\}_{i=1}^{n_A}$ and $B=\{\hat B_j\}_{j=1}^{n_B}$, respectively, the tensor-product majorization relation introduced in \cite{Friedland2013,Puchala2013} turns out to be
\begin{align}
p^A \otimes p^B \prec w_t,
\end{align}
where $p^A \otimes p^B=(p^A_1 p^B_1,...,p^A_1 p^B_{n_B},...,p^A_{n_A} p^B_{1},...,p^A_{n_A} p^B_{n_B})^T$ is the $n_An_B$-dimensional joint probability vector.
Here, {the} majorizing vector $w_t$ is given by
\begin{align}\label{wt}
w_t=\left(\frac{s_2^2} {4},\frac{s^2_3-s^2_2} {4},...,\frac{s_N^2-s_{N-1}^2} {4},0,...,0\right)^T
\end{align}
with the total number of measurement outcomes $N=n_A+n_B$ and the {coefficients}
\begin{align}\label{s}
s_k:= \max_{\substack{ \mathcal R, \mathcal S\\  |\mathcal R|+|\mathcal S|=k }}\left\|{\sum_{i\in \mathcal R}\hat A_i}+{\sum_{j\in \mathcal S}\hat B_j}\right\|.
\end{align}
Here, $\|\cdot\|$ is the operator norm equal to {the} largest singular value, and $\mathcal R \subset\{1,...,n_A\}$ and $\mathcal S \subset\{1,...,n_B\}$, with $|\mathcal R|$ indicating the number of elements of $\mathcal{R}$.
Because of the additivity of entropic quantities, it is  straightforward to derive  {EURs} in terms of the R\'enyi  entropy~as
\begin{align}\label{tensorEUR}
&H_\alpha(A)+H_\alpha(B)\geq H_\alpha(w_t)\equiv \mathcal B_t.
\end{align}
Note that the bound $\mathcal B_t$ is determined only by the considered  {POVMs $A$ and $B$, which give} a state-independent bound in Equation \eqref{tensorEUR}.

More recently, the direct-sum majorization relation was considered for the case of rank-1 PVMs in \cite{Rudnicki2014} and generalized to POVMs in \cite{Rastegin2016}. To address this approach, let us introduce  {the $n_A d \times n_B d$ block matrix $\mathsf{X}$ consisting of $d \times d$ blocks $\mathsf{X}_{ij} = \sqrt{\hat A_i} \sqrt{\hat B_j}$, given by}
\[\mathsf X =
\begin{pmatrix}
\sqrt{\hat A_1}\sqrt{\hat B_1}  & \cdots & \sqrt{\hat A_1}\sqrt{\hat B_{n_B}} \\
\vdots    & \ddots & \vdots  \\
\sqrt{\hat A_{n_A}}\sqrt{\hat B_{1}} & \cdots & \sqrt{\hat A_{n_A}}\sqrt{\hat B_{n_B}}
\end{pmatrix}.\]
{Here and going forward, $d$ is the dimension of the Hilbert space of the system.}
This matrix includes all combinations of POVM elements between $A$ and $B$.
We also define the set of {block} submatrices such that
\begin{equation}
\begin{aligned}
\mathcal{SUB}(\mathsf X,k)=&\{\mathsf Z\in \mathcal{M}_{rd\times r'd}(\mathbb{C}):\\
&\text{$\mathsf Z$ is a submatrix of  {$\mathsf X$ made up of $d\times d$ blocks.}}, \\
& r+r'-1=k \} ,
\end{aligned}
 \end{equation}
where $\mathcal{M}_{rd\times r'd}(\mathbb{C})$ denotes the space of all $rd \times r'd$ complex matrices, {and $1\leq k\leq n_A+n_B-1$ is a positive integer}.

With the above definitions, the direct-sum majorization relation turns out to be {\mbox{\cite{Rudnicki2014,Rastegin2016}}}
\begin{align}\label{Direct0}
p^A\oplus p^B\prec (1)\oplus w_d,
\end{align}
where $p^A\oplus p^B=(p_1^A,...,p_{n_A}^A,p_1^B,...,p_{n_B}^B)$. Here, the majorizing vector $w_d$ is given by
\begin{align}
w_d=(c_1,c_2-c_1,...,c_{N-1}-c_{N-2})^T
\end{align}
{with the} coefficients
\begin{align}\label{c}
c_k=\max\{\|\mathsf Z\|: \mathsf Z\in \mathcal{SUB}(\mathsf X,k)\}.
\end{align}
{It is worth noting that the majorization relation is applicable to unnormalized nonnegative vectors if the sum of the vector components has the same value; for instance, the components of each vector in Equation {\eqref{Direct0}} sum to 2. Furthermore, for a pair of vectors with different lengths, one can adjust} {the majorization relation} {by adding zeros to additional coordinates of the shorter vector, such as $w_t$ in Equation {\eqref{wt}}.}

The direct-sum majorization relation allows one to derive the following {EURs} \cite{Rudnicki2014,Rastegin2016}.  {For R\'enyi entropies of order $0 < \alpha \leq 1$},  {we have}
\begin{align}\label{REURwd}
H_\alpha(A)+H_\alpha(B)\geq H_{\alpha}(w_d)\equiv \mathcal B_{d1}.
\end{align}
For $\alpha>1$, they satisfy another form of inequalities:
\begin{align}
H_\alpha(A)+H_\alpha(B)\geq \frac{2} {1-\alpha}\ln \left(\frac{1} {2}+\frac{1} {2}\sum_{i=1}^d (w_d)_i^\alpha\right).
\end{align}
For Tsallis entropies  {of any order $\alpha>0$, a relation of the same form as Equation {\eqref{REURwd}} turns out to~be}
\begin{align}
T_\alpha(A)+T_\alpha(B)\geq T_{\alpha}(w_d).
\end{align}

{The} approach to quantify the incompatibility {based on the matrix $\mathsf{X}$} was introduced {for a pair of rank-1 PVMs in \mbox{\cite{Puchala2013,Rudnicki2014}}. Therein, the matrix $\mathsf{X}$ is just the unitary matrix connecting the two orthonormal bases associated with the respective PVMs. Subsequently, in {\cite{Rastegin2016}}, this approach was extended to POVMs by means of the matrix $\mathsf{X}$ defined above. Since this quantification takes into account many different operator norms defined by $c_k$ in Equation {\eqref{c}}, it is an} extension of the Maassen--Uffink bound that is only determined by the largest operator norm \cite{Krishna2002}:
\begin{align}\label{MU}
\mathcal B_{MU}=-2 \ln \max_{ij} \left\| \sqrt{\hat A_i}\sqrt{\hat B_j} \right\|.
\end{align}
{However, this extension does not always provide stronger bounds than $\mathcal B_{MU}$. Despite its simple and intuitive form, the Maassen--Uffink bound is complementary to the majorization EURs, particularly for mutually unbiased bases. Furthermore, $\mathcal B_{MU}$ was improved on the basis of the Landau--Pollak inequality for rank-1 PVMs in \mbox{\cite{Vicente2008,Bosyk2011}}.}  {The improved bounds contained therein were subsequently} {extended to the case of POVMs in \mbox{\cite{Bosyk2014,Zozor2014}}.}

{For the case of rank-1 PVMs, the coefficients $s_k$ and $c_k$ are related by the following equality {\cite{Puchala2013}}}: $$s_{k+1}=1+c_k$$
{for $k=1,2,...,d$. This relation allows us to analytically compare the majorizing vectors $w_{t}$ and $w_{d}$} {since it gives} {$\sum_{i=1}^k (w_{t})_i=(1+c_k)^2/4 \geq \sum_{i=1}^k (w_{d})_i= c_k$ as a result of the inequality of arithmetic and geometric means. Note that only this inequality does not imply majorization since} {$w_d$} {is not sorted in decreasing order.}  {The following majorization relation was rigorously proved in \cite{Rudnicki2014}:}
\begin{align}\label{MajorDirectTensor}
w_d \prec w_t.
\end{align}
{This implies that the direct-sum majorization relation gives improved bounds for rank-1 PVMs.} Thus, we have $H_\alpha(w_d)\geq H_\alpha(w_t)$ for $0<\alpha\leq 1$ and $T_\alpha(w_d)\geq T_\alpha(w_t)$ for $\alpha>0$. However, this improvement is not observed in the generalization to POVMs \cite{Rastegin2016} {(see Section {\ref{sec4}} for extensive investigations in the qubit case)}. Thus, our main purpose is to find {some} new generalization of the direct-sum majorization relation that also gives {an} improvement relative to the existing EURs in the POVM~case.


\section{Direct-Sum Majorization Relations for General POVM}\label{sec3}

In this section, we suggest a new generalization of the direct-sum majorization relation to POVMs. By utilizing it, we further derive {EURs} for R\'enyi and Tsallis entropies.
For the case of rank-1 PVMs, the direct-sum majorization relation was derived in \cite{Rudnicki2014} {and extended to PVMs {\cite{Rudnicki2015}} and POVMs {\cite{Rastegin2016}}}. The main idea of the derivation of the direct-sum majorization relation {in {\cite{Rudnicki2014}}} is to find {the} majorizing vector by taking the largest operator norm {of sums of rank-1 PVM elements. We} apply this idea to the case of POVMs as follows.

\begin{Theorem}
	For POVMs $A$ and $B$, we have the majorization relation
	\begin{align}\label{Direct}
	p^A\oplus p^B  \prec  W
	\end{align}
	where $p^A$ and $p^B$ are the probability vectors whose elements are defined as $p^A=(p^A_1,p^A_2,...,p^A_{n_A})^T${}, $p^B=(p^B_1,p^B_2,...,p^B_{n_B})^T${}, and {the} $N$-dimensional vector $W$ is defined as
	\begin{align}\label{W}
	W=(s_1,s_2-s_1,...,s_N-s_{N-1})^T
	\end{align}
	with $N=n_A+ n_B$.
\end{Theorem}
Note that the {coefficients} $s_j$ $(j=1,\cdots,N)$ in Equation \eqref{W} are the same as those defined in Equation \eqref{s}.
\begin{proof}[Proof of Theorem \mbox{{1}}]
	Let us {assume that $z_\downarrow=(z_{[1]},z_{[2]},...,z_{[N]})^T$ is the rearrangement of  $z=p^A\oplus p^B$ in decreasing order.}  Then, the largest element of $z_\downarrow$ is {either} $p^A_{[1]}$ or $p^B_{[1]}$. In each case, we have inequalities
	\begin{align}
	p^A_{[1]}=\text{tr}[\hat\rho \hat A_{[1]}]\leq \max_i [\|\hat A_i\|]\nonumber,\\
	p^B_{[1]}=\text{tr}[\hat\rho \hat B_{[1]}]\leq \max_j [\|\hat B_j\|]\nonumber,
	\end{align}
	for any density operator $\hat\rho$ as a result of {the definition of the operator norm, $\|\hat A_i\|=\max_{|\psi\rangle} \langle \psi|\hat A_i|\psi\rangle$.} Here, the upper bounds are lower than 1 by the definition of a POVM.  Then, the summation of the first and the second largest element{s} of $z_\downarrow$ has three possible combinations in each case satisfying following~ inequalities
	\begin{align}
	p^A_{[1]}+p^A_{[2]}&=\text{tr}[\hat\rho(\hat A_{[1]}+\hat A_{[2]})]\leq \max_{\substack{ \mathcal R\\  |\mathcal R|=2 }}\left\|\sum_{i\in \mathcal R}\hat A_i\right\|\nonumber,\\
	p^A_{[1]}+p^B_{[1]}&=\text{tr}[\hat\rho(\hat A_{[1]}+\hat B_{[1]})]\leq \max_{\substack{ \mathcal R, \mathcal S\\  |\mathcal R|=|\mathcal S|=1 }}\left\|\sum_{i\in \mathcal R}\hat A_i+\sum_{j\in \mathcal S}\hat B_j\right\|\nonumber,\\
	p^B_{[1]}+p^B_{[2]}&=\text{tr}[\hat\rho(\hat B_{[1]}+\hat B_{[2]})]\leq \max_{\substack{ \mathcal S\\  |\mathcal S|=2 }}\left\|\sum_{j\in \mathcal S}\hat B_j\right\|\nonumber,
	\end{align}
	where $\mathcal R \subset\{1,...,n_A\}$ and $\mathcal S \subset\{1,...,n_B\}$. {Each inequality comes from the definition of the operator norm.} In the same way, the summation of $k$ elements of vectors has an upper bound as follows:
	\begin{align}
	\sum_{i=1}^k z_{[i]}\leq  \max_{\substack{ \mathcal R, \mathcal S\\  |\mathcal R|+|\mathcal S|=k }}\left\|\sum_{i\in \mathcal R}\hat A_i+\sum_{j\in \mathcal S}\hat B_j\right\|=s_k\nonumber.
	\end{align}
Because of the completeness relation, we have $s_{N}=2$.

	Now, we can construct {the} $N$-dimensional majorizing vector
	$W=(s_1,s_2-s_1,s_3-s_2,...,s_{N}-s_{N-1})^T$ that satisfies the direct-sum majorization relation
	\begin{align}
	\sum_{i=1}^k z_{[i]}\leq \sum_{i=1}^k W_i
	\end{align}
	for arbitrary $k$.
	\end{proof}

{The majorizing vector $W$} {coincides with the vector} {$(1)\oplus w_d$ for the case of rank-1 PVMs since $s_{k+1}=1+c_{k}$, as derived in {\cite{Puchala2013}}, together with $s_1=1$.} {However, for general POVMs, the equality is replaced with the inequality, i.e., $s_{k+1}\leq1+c_{k}$.} {This fact implies that distinct behaviors of $W$ from $(1)\oplus w_d$ may be observed for unsharp observables that cannot be described by PVMs.}
Significant distinctions between the direct-sum majorization relation in Equation \eqref{Direct} and the previous one in Equation \eqref{Direct0} are encapsulated in the following relation {{\cite{Rastegin2016}}}:
\begin{align}\label{MajVectos}
\sum_{i=1}^k z_{[i]}\leq \sum_{i=1}^k W_i \leq \sum_{i=1}^k \left( (1)\oplus w_d\right)_i
\end{align}
for a given  {$k\in[1,N]$}, where {$z_\downarrow=(z_{[1]},z_{[2]},...,z_{[N]})^T$ is} {the} {vector $z=p^A\oplus p^B$ ordered decreasingly, and} $(1)\oplus w_d=(1,c_1,c_1-c_2,...,c_{N-1}-c_{N-2})^T$.  We note that for any pair of $A$ and $B$, one can find a state saturating the first inequality in Equation  \eqref{MajVectos} for  {each} $k$, since {$\sum_{i=1}^k W_i=s_k$} is defined {by} taking the largest {eigenvalue of all possible sums of $k$ POVM	elements.}

It is worth noting that the inequalities in Equation  \eqref{MajVectos} were actually mentioned in Reference \cite{Rastegin2016}, but it was further claimed that equality holds in the second inequality. However, there are cases {where} $w_d$ does not coincide with the majorizing vector $W$, but they become equivalent for rank-1 PVMs. For instance, let us consider qubit observables $A=\{(\hat I\pm \mu \hat\sigma_x)/2\}$ and $B=\{(\hat I\pm \mu \hat\sigma_z)/2\}$, {where $\hat \sigma_x$ and $\hat \sigma_z$ denote the Pauli matrices and $0 \leq\mu\leq 1$ is an unsharpness parameter.} The majorizing vector given by  {$W=((1+\mu)/2,(1+(\sqrt{2}-1)\mu)/2,(1-(\sqrt{2}-1)\mu)/2,(1-\mu)/2)^T$}
has a specific state for each $k$ saturating the first inequality, while
$(1)\oplus w_d=(1,c_1,c_2-c_1,c_3-c_2)^T$ with $c_1=(\mu+\sqrt{2-\mu^2})/2\sqrt{2}$, $c_2=\sqrt{(1+\mu)/2}$, and $c_3=1$ {in {\cite{Rastegin2016}}} does not.
More details about {the} difference between the two majorizing vectors $W$ and $ (1)\oplus w_d$ are given in Section \ref{sec4} by explicitly showing that our EUR performs better than the previous one.
{However, we note that one cannot  infer $W\prec(1)\oplus w_d$ from  Equation {\eqref{MajVectos}}, since the vectors $W$ and $(1)\oplus w_d$ are not sorted in decreasing order.}

The direct-sum majorization relation in Equation \eqref{Direct} allows one to {derive the following EURs by means of the mathematical techniques employed in {\cite{Rudnicki2014}}.}

\begin{Corollary}
	For a pair of POVMs $A$ and $B$, we have {the following} entropic uncertainty relations for R\'enyi entropies of order $0<\alpha\leq1$:
	\begin{align}\label{MajPOVM1}
	H_\alpha(A)+H_\alpha(B)\geq \frac{1} {1-\alpha} \ln \left(\sum_{i=1}^N (W_i)^\alpha-1\right)\mbox{{$\equiv\mathcal B_{d2}$}},
	\end{align}
	and for $\alpha>1$
	\begin{align}\label{MajPOVM2}
	H_\alpha(A)+H_\alpha(B)\geq \frac{2} {1-\alpha} \ln \left(\frac{\sum_{i=1}^N (W_i)^\alpha} {2}\right).
	\end{align}
\end{Corollary}
\begin{proof}[Proof of Corollary 1]
	First, for the case $0<\alpha< 1$, the authors in \cite{Rudnicki2014} found that
	\begin{align}
	H_\alpha(A)+H_\alpha(B)\geq \frac{1} {1-\alpha} \ln \left(\sum_i (p_i^A)^\alpha+\sum_j (p_j^B)^\alpha-1\right).\nonumber
	\end{align}
Because of the Schur concavity of $\sum_i x_i^\alpha$ for $\alpha<1$, one can obtain the inequality in Equation \eqref{MajPOVM1} by using the direct-sum majorization relation we provide {in Equation {\eqref{Direct}}}. For the case $\alpha=1$, the left-hand side can be written as $-\sum_{i=1}^N z_i \ln z_i$. By applying the Schur concavity of that function, we {obtain the following bound with a form similar to the Shannon entropy}:
	\begin{align}\label{MajPOVM3}
	H_\alpha(A)+H_\alpha(B)\geq-\sum_{i=1}^N W_i\ln W_i.
	\end{align}
We note that the bound $\mathcal B_{d2}$ {reduces} to  {a form similar to the Shannon entropy in Equation {\eqref{MajPOVM3}} in the limit $\alpha \rightarrow 1$}.
	
	For the case $\alpha>1$, by applying a relation between geometric and arithmetic means, we have \cite{Rudnicki2014}
	\begin{align}
	H_\alpha(A)+H_\alpha(B)\geq \frac{2} {1-\alpha} \ln \left(\frac{\sum_i (p_i^A)^\alpha+\sum_j (p_j^B)^\alpha} {2}\right).\nonumber
	\end{align}
By using the fact that the bound is Schur concave, we can find straightforwardly the inequality~in Equation \eqref{MajPOVM2} from the direct-sum majorization relation.
\end{proof}

For the Tsallis entropy of any order $\alpha>0$, the direct-sum majorization relation {yields a unified formula for EURs}  as follows.
\begin{Corollary} For the Tsallis entropy of any order $\alpha>0$, we have
	\begin{align}\label{TEUR}
	T_\alpha(A)+T_\alpha(B)\geq \frac{1} {1-\alpha}\left(\sum_{i=1}^n W_i^\alpha -2 \right)
	\end{align}
\end{Corollary}
\begin{proof}[Proof of Corollary 2]
	By the definition of the Tsallis entropy, the left-hand side of Equation \eqref{TEUR} can be written as
	\begin{align}
	T_\alpha(A)+T_\alpha(B)=\frac{1} {1-\alpha}\left(\sum_{i=1}^n z_i^\alpha -2 \right).
	\end{align}
By using the fact that $(\sum_i x_i^\alpha)/(1-\alpha)$ is Schur concave for $\alpha>0$, we obtain the inequality in Equation \eqref{TEUR} from the direct-sum majorization relation.
\end{proof}

\section{Comparison of Bounds}\label{sec4}

In this section, we compare the bound derived from the direct-sum majorization relation in Equation {{\eqref{Direct}}} with the previous bounds for the sum of two Shannon entropies.

\subsection{Qubit Observables}

 {As the simplest nontrivial example, let us consider a pair of qubit observables. Indeed, in this framework, the optimal EURs  have been established for the case of rank-1 PVMs in terms of Shannon~\mbox{\cite{Sanches1998,Ghirardi2003}} and R\'enyi entropies \mbox{\cite{Zozor2013}}. However, it has not been studied intensively for the case of POVMs. Thus, the goal of this section is to illustrate our bound by showing how it works for unsharp qubit observables in comparison with others. Now, let us define the following POVMs $X(\theta)$ and $Z$ }:
\begin{align}
\hat{X}_\pm(\theta)=\frac{\hat I\pm\mu(\sin\theta \hat \sigma_x+\cos\theta \hat{\sigma}_z )} {2},\;\; \hat{Z}_\pm=\frac{\hat I\pm\nu\hat \sigma_z } {2}, \nonumber
\end{align}
where $\theta$ refers to the angle between measurement directions, while $\mu$ and $\nu$ determine the unsharpness of  {measurements $X(\theta)$ and $Z$}, respectively. In this case, the majorizing vector $W$ in Equation \eqref{Direct} is given by
\begin{align}
s_1&=\frac{1+\max[\mu,\nu]} {2} ,\nonumber \\
s_2&= 1+ \frac{1} {2}\sqrt{\mu^2+\nu^2+2\mu\nu |\cos \theta|}, \nonumber \\
s_3&= \frac{3+\max[\mu,\nu]} {2}. \nonumber
\end{align}

In Figure \ref{Qubit}, all bounds derived via the majorization technique {are compared with the Maassen--Uffink bound for varying angles} $\theta$ and unsharpness parameters $\mu=\nu$. First, in Figure \ref{Qubit}a, we plot those bounds at $\mu=1$, i.e., for the case of rank-1 PVMs. This plot illustrates that our bound $\mathcal B_{d2}$ reproduces the direct-sum majorization bound $\mathcal B_{d1}$ as {claimed after Theorem 1}. In $|\theta-\pi/2|>0.15$, $\mathcal B_{d2}$ is stronger than $\mathcal B_{MU}$, while it is weaker in the other region. On the other hand, Figure \ref{Qubit}b shows that $\mathcal B_{d2}$ is the most refined bound {for the fixed unsharpness parameter} $\mu=0.8$. Our bound $\mathcal B_{d2}$ tends to be a stronger bound than others with increasing uncertainty due to measurement unsharpness. This is more clearly shown in Figure \ref{Qubit}c,d, where we plot all bounds versus $\mu$ at fixed $\theta=\pi/2, \pi/3$, respectively.  In Figure \ref{Qubit}c, $\mathcal B_{d2}$ is stronger than $\mathcal B_{MU}$ when $\mu<0.967$.
 {Performing the unsharp measurement $Z$ is equivalent to the case where the Pauli measurement $\sigma_z$ is performed with white noise amounting to $1-\nu$} {\mbox{\cite{Busch1986,Heinosaari2015}}}, {and likewise for $X(\theta)$. Therefore, $\mathcal B_{d2}$ provides a stronger bound for the case where there exists an amount} {$1-\mu=1-\nu$} {of white noise larger than 0.033 for $\theta=\pi/2$.}
Furthermore, in the case of $\theta=\pi/3$, $\mathcal B_{d2}$ provides the most refined bound for all values of $\mu$, {as illustrated in Figure {\ref{Qubit}}d}.

\begin{figure}[t]
	\centering \includegraphics[width=0.8\textwidth]{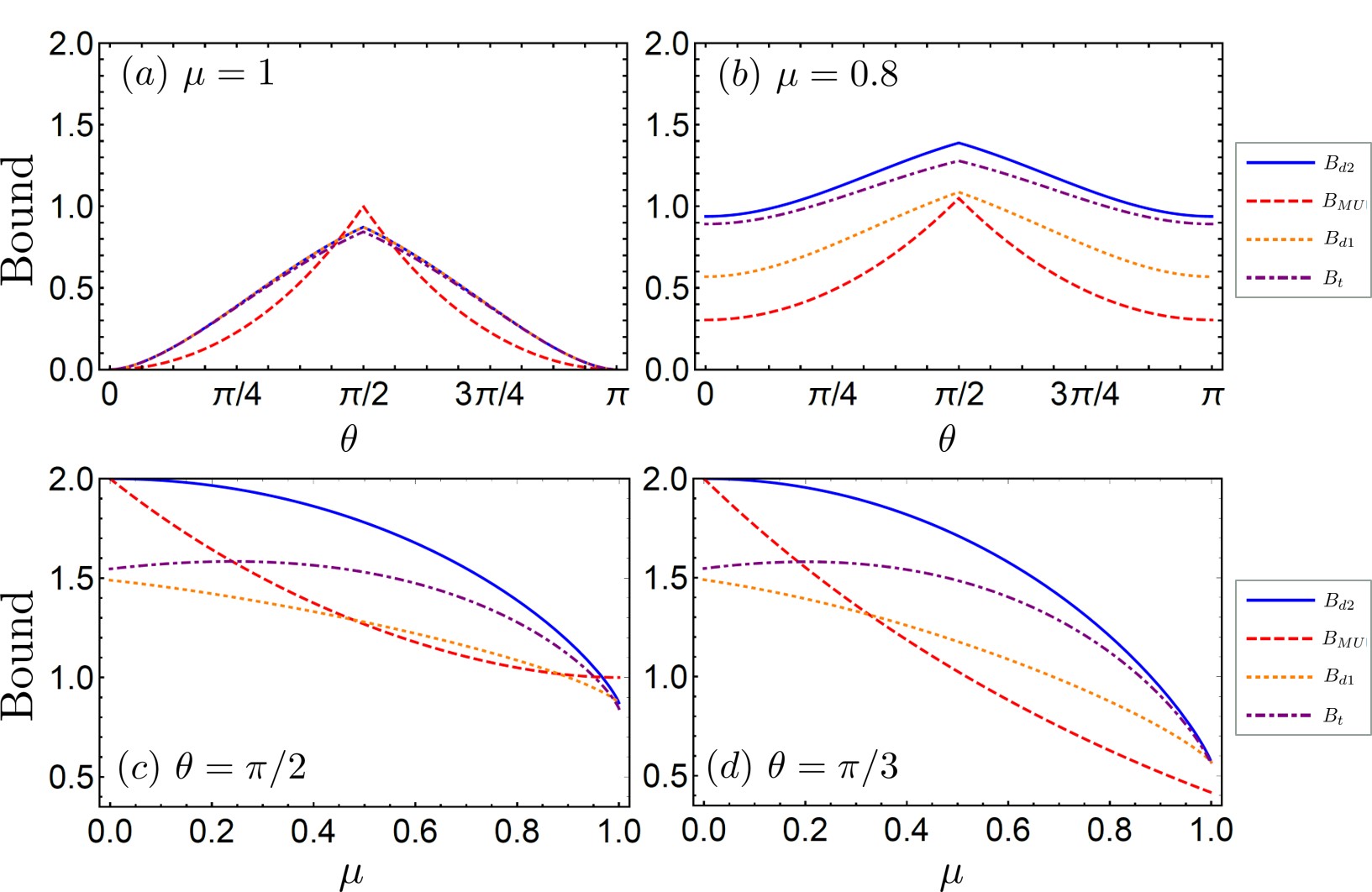}
	\caption{Bounds for the sum of two Shannon entropies rescaled to the logarithm with base 2.
	(\textbf{a},\textbf{b})~Plots of the bounds versus the angle $\theta$ at fixed unsharpness parameters (\textbf{a}) $\mu=1$ and (\textbf{b}) $\mu=0.8$; (\textbf{c},\textbf{d}) Plots of the bounds versus the unsharpness parameter $\mu$ at fixed angles (\textbf{c}) $\theta=\pi/2$ and (\textbf{d}) $\theta=\pi/3$.
	(Blue solid curves: our direct-sum majorization bound $\mathcal B_{d2}$ in Equation \eqref{MajPOVM1}; red dashed curves: Maassen--Uffink bound $\mathcal B_{MU}$ in Equation \eqref{MU}; orange dotted curves: previous direct-sum majorization bound $\mathcal B_{d1}$ in Equation \eqref{REURwd}; and purple dot-dashed curves:  tensor-product majorization bound $\mathcal B_{t}$ in Equation \eqref{tensorEUR}).}\label{Qubit}
\end{figure}

\subsection{High-Dimensional System}

As a nontrivial example in three-dimensional {systems}, let us consider orthogonal bases $\{|1\rangle, |2\rangle,|3\rangle\}$, and $\{\hat U|1\rangle, \hat U|2\rangle,\hat U|3\rangle\}$, with
\[\hat U=\left(
\begin{array} {ccc}
1/\sqrt{3}&1/\sqrt{3}&1/\sqrt{3}\\
1/\sqrt{2}&0&-1/\sqrt{2}\\
1/\sqrt{6}&\mbox{$-$}\sqrt{2/3}&1/\sqrt{6}
\end{array}
\right),\]
which was used to examine the quality of various bounds for rank-1 PVMs in \cite{Coles2014,Rudnicki2014}.
Furthermore, to apply it to the case of POVMs, we apply randomly generated $3\times3$ doubly stochastic matrices, $S^f$ and $S^g$, to each observable {so that}
\begin{align}
\hat F_i&=\sum_{k=1}^3 S^f_{ik} |k\rangle\langle k|,\\
\hat G_j&=\sum_{l=1}^3 S^g_{jl} \hat U|l\rangle\langle l|\hat U^\dagger,
\end{align}
for $i,j\in\{1,2,3\}$, which are elements of the {POVMs} $F$, $G$, respectively.

\begin{figure}[t]
	\centering \includegraphics[width=0.7\textwidth]{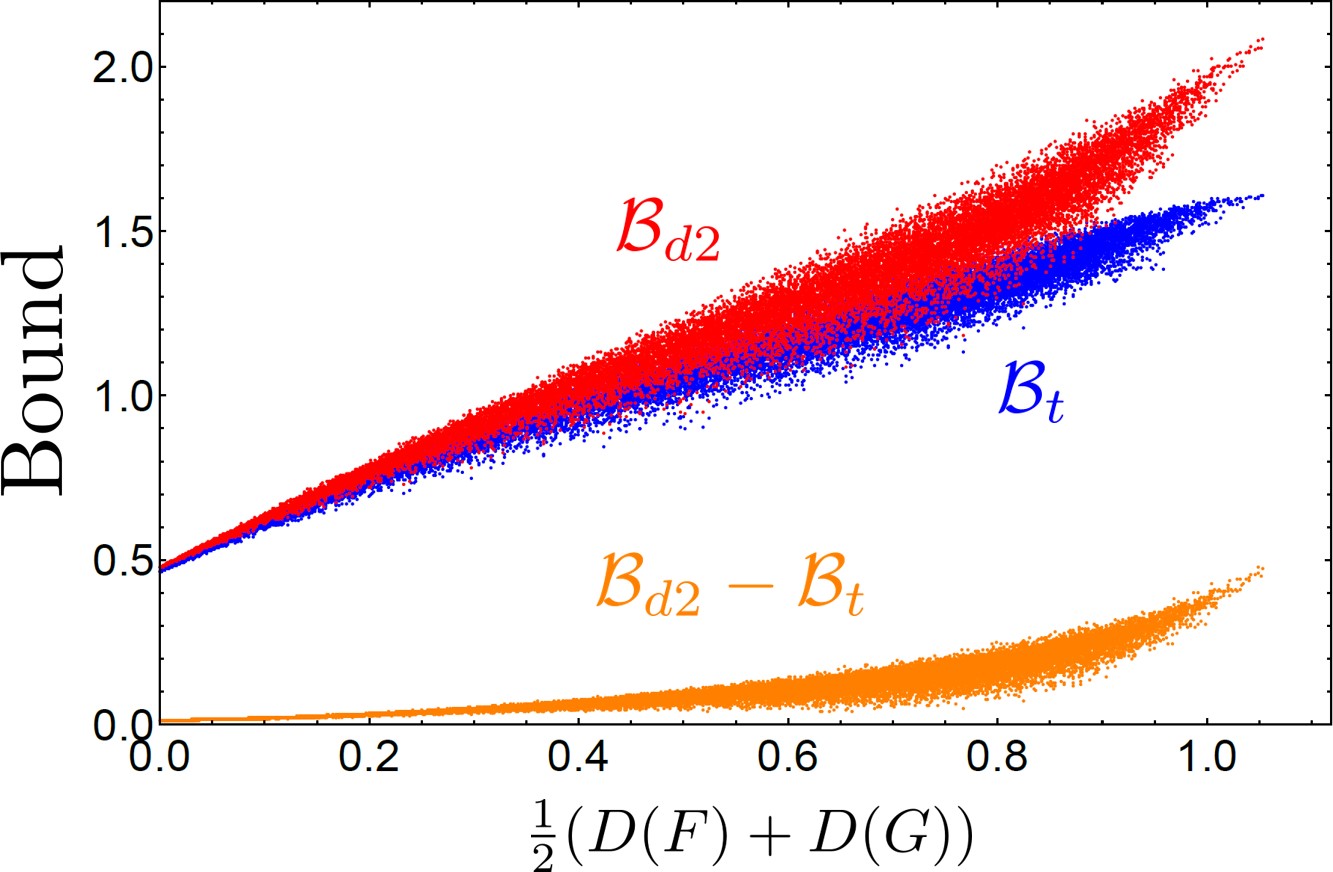}
	\caption{Plot of $\mathcal B_{d2}$ (red), $\mathcal{B}_t$ (blue), and their difference $\mathcal B_{d2}-\mathcal B_t$ (orange) versus the averaged device uncertainty $\frac{1} {2}(D(F)+D(G))$. The logarithm is taken with respect to the base $e$, where the most trivial measurement case, i.e., $\hat F_i=\hat G_j=\hat I/3$ for all $i$, $j$, coincides with the point $\frac{1} {2}(D(F)+D(G))=\ln3\sim 1.1$.}\label{POVM}
\end{figure}

For the case of rank-1 PVMs, it has been verified that $\mathcal B_{d2}$ is stronger than $\mathcal B_t$ because of the relation~in Equation \eqref{MajorDirectTensor}. To {provide numerical examples} illustrating that it is also the case for general POVMs, we compare $\mathcal B_{d2}$ with $\mathcal B_t$  for randomly generated doubly stochastic matrices in Figure \ref{POVM}. {To clearly show their difference,}  {we also exhibit the difference $\mathcal B_{d2}-\mathcal B_{t}$.}
We plot these values versus the degree of unsharpness quantified by so-called device uncertainty \cite{Baek2016-1,Baek2016-2}, $D(F)=-\text{Tr}[(\hat I/3) (-\sum_{i=1}^3 \hat F_i \ln \hat F_i)]=-(1/3)(\sum_{i,k=1}^3 S^f_{ik}\ln S^f_{ik})$, and likewise for $D(G)$. We see that $\mathcal B_{d2}$ gives better bounds than $\mathcal B_t$, as illustrated by {the difference} $\mathcal B_{d2}-\mathcal B_t$ {being} {positive} in all cases. The gap tends to be larger as the degree of unsharpness increases. This result provides evidence that $\mathcal B_{d2}$ provides a stronger bound than $\mathcal B_t$ as expected.


\section{Multiple Measurements}\label{sec5}

One of the important advantages of the direct-sum majorization relation derived in Section \ref{sec3} is that it can be readily generalized to the case of an arbitrary number of $L$ measurements. For the case of multiple rank-1 PVMs, this generalization was made in \cite{Rudnicki2014}. In this section, we provide its generalization to the case of multiple POVMs.

Let us consider a collection of $L$ measurements $\{M_l\}_{l=1}^{L}$, where each measurement is described by its component operators $\{\hat M_{i|l}\}_{i=1}^{n_l}$. The probability distribution associated with {the} $l$th measurement is written as a column vector $P_{l}=(P_{1|l},...,P_{n_l|l})^T$, where $P_{i|l}=\text{Tr}[\hat\rho\hat M_{i|l}]$. With this notation, we can show the following results.

\begin{Theorem}
	For POVMs $\{M_l\}_{l=1}^{L}$, we have the majorization relation
	\begin{align}\label{DirectMulti}
	Z=\bigoplus_{l=1}^L P_l  \prec  W
	\end{align}
	where the $N$-dimensional vector $W$ is defined as
	\begin{align}
	W=(\mathcal S_1,\mathcal S_2-\mathcal S_1..., \mathcal S_N-\mathcal S_{N-1})^T
	\end{align}
	with $N=\sum_{l=1}^L n_l$, where
	\begin{align}
	\mathcal S_k:= \max_{\substack{ \mathcal R_l\\  \sum_{l=1}^L|\mathcal R_l|=k }}\left\|\sum_{l=1}^L\sum_{i\in \mathcal R_l}\hat M_{i|l}\right\|
	\end{align}
	with $\mathcal R_l \subset\{1,...,n_l\}$.
\end{Theorem}
\begin{proof}[Proof of Theorem 2]
	Let us define the $N$-dimensional column vector  {$Z_\downarrow=(Z_{[1]},Z_{[2]},...,Z_{[N]})^T$, which is the rearrangement of $Z:=\bigoplus_{l=1}^L P_l$ in decreasing order.} Without loss of generality, we {let $Z_{[j]} = P_{[i_j|l_j]}$ be the $j$th element of $Z_\downarrow$.} With these definitions, we can show that the sum of $Z_{[j]}$ up to the $k$th element has the upper~bound
	\begin{align}
	\sum_{j=1}^k Z_{[j]}&=\sum_{j=1}^{k} P_{[i_j|l_j]}=\sum_{j=1}^k\text{Tr}[\hat \rho \hat M_{[i_j|l_j]}]\nonumber\\
	&\leq\left\|\sum_{j=1}^{k} \hat M_{i_j|l_j}\right\|\leq\max_{\substack{ \mathcal R_l\\  \sum_{l=1}^L|\mathcal R_l|=k }}\left\|\sum_{l=1}^L\sum_{i\in \mathcal R_l}\hat M_{i|l}\right\|=\mathcal S_k\nonumber.
	\end{align}
{Similar to the proof of} Theorem 1, in the first inequality, we use the property of the operator norm, and in the second inequality, we use the fact that {$\mathcal S_k$} is obtained by finding the maximum operator norm over all {combinations} of POVM elements.
\end{proof}

The direct-sum majorization relation for multiple measurements allows us to derive EURs in terms of the Shannon entropy,
\begin{align}
\sum_{l=1}^L H(M_l) \geq -\sum_{i=1}^N W_i \ln W_i.
\end{align}
Furthermore, as noted in \cite{Rudnicki2014}, in the case of R\'enyi entropies with $\alpha<1$,  one can have
\begin{align}
\sum_{l=1}^L H_\alpha(M_l) \geq \frac{1} {1-\alpha}\ln\left(\sum_{i=1}^N (W_i)^\alpha+1-L\right).
\end{align}
by using the same method applied in the derivation of Equation \eqref{MajPOVM1}.
Also, in the case of Tsallis entropies, it is straightforward to obtain
\begin{align}
\sum_{l=1}^L T_\alpha(M_l) \geq  \frac{1} {1-\alpha}\left(\sum_{i=1}^n W_i^\alpha -L \right).
\end{align}

In the case of multiple projective measurements, the bound obtained via the direct-sum majorization relation was shown to be nontrivial in comparison with others, {as examined} in \cite{Rudnicki2014,Liu2015}.
This also implies that our method can provide significantly useful bounds in the case of multiple generalized measurements, {because our generalization includes the previous result in {\cite{Rudnicki2014}} as a particular case.}


\section{Conclusions}\label{Con}

In this work, we provide the direct-sum majorization relation for generalized measurements in Equation \eqref{Direct}. As an extension of the approach in \cite{Rudnicki2014} to general POVM measurements, our direct-sum majorization relation reproduces the result of projective measurements as a special case. Furthermore, we show that our method  {yields the majorizing vector} in Equation {{\eqref{W}}}, {which is a significant improvement of the one presented in {\cite{Rastegin2016}}.}

On the basis of this direct-sum majorization, we established EURs for R\'enyi and Tsallis entropies, including the Shannon entropy. To illustrate the usefulness of our EURs,  {in the case of two POVMs, we compared our Shannon entropy UR with other known similar EURs.} First, for qubit observables, we show that our bound is stronger than other majorization bounds, while it can be complementary to the Maassen--Uffink bound. Our bound provides a significant improvement relative to  {other bounds}, particularly when the {measurement unsharpness} is significant. Secondly, in three-dimensional systems, we considered a pair of unsharp measurements generated by randomly mixing  {two different orthogonal bases.} We obtained numerical evidence exhibiting that our bound derived from direct-sum majorization is stronger than the one from the tensor-product in~\cite{Friedland2013}. {Our finding significantly extends the result known for the case of projective measurements, as rigorously proved in {\cite{Rudnicki2014}}, to general POVMs.} {Our result significantly extends the one proved in {\cite{Rudnicki2014}} from the case of projective measurements to general POVMs.}

We further extended our approach to the case of multiple POVMs via {a} direct-sum majorization relation that allows us to achieve new bounds for R\'enyi and Tsallis entropies. This extension is useful for exploring URs for the most general measurement scenario, which has so far not been studied extensively compared with the multiple projective measurements scenario. As a future work, we may establish EURs by incorporating information on the mixedness of the state to obtain a tighter bound for the case of mixed states. The recent work in \cite{Puchala2018} considered such a problem for the case of projective measurement on the basis of the idea of state purification, which can be further extended to POVM measurements as well. More broadly, it may be interesting to extend our approach to bipartite systems in which entanglement can act as a resource to reduce the amount of uncertainty in the measured~system.

\vspace{6pt}



\noindent{\bf Author contributions:} K.B. conceived the problem, which was theoretically developed together with H.N. and W.S. All authors contributed to writing and proofreading the manuscript.\\

\noindent {\bf Funding: }This work is supported by the NPRP grant 8-751-1-157 from Qatar National Research Fund.\\

\noindent {\bf Abbreviations: }{The following abbreviations are used in this manuscript:\\

\noindent
\begin{tabular} {@{}ll}
{UR} & {Uncertainty Relation}\\
EUR & Entropic Uncertainty Relation\\
POVM & Positive-Operator-Valued Measure\\
PVM & Projection-Valued Measure
\end{tabular}}


\end{document}